\documentclass[a4paper]{IEEEtran}
\usepackage{amssymb,amsthm,amsmath}
\usepackage[utf8]{inputenc}

\newtheorem{theorem}{Theorem}
\newtheorem{lemma}[theorem]{Lemma}
\newtheorem{claim}[theorem]{Claim}
\newtheorem{proposition}[theorem]{Proposition}

\def\eps{\varepsilon}
\def\ITEM#1{\par\smallskip\noindent\hangindent=1.7\parindent\hangafter1\hbox to
1.7\parindent{\hss\upshape#1\space}\ignorespaces}
\newcommand\M{{\mathcal M}} 
\newcommand\N{{\mathcal N}} 
\newcommand\G{{\mathcal G}} 
\newcommand\F{{\mathcal F}} 
\newcommand\cl{\mathop{\mbox{\sf cl}}}
\newcommand\contract{\!\mathbin{\not{~}}}

\title{Sticky polymatroids on at most five elements}
\author{Laszlo Csirmaz%
\thanks{Central European University, Budapest}%
\thanks{email:\tt csirmaz@renyi.hu }}
\IEEEspecialpapernotice{%
Dedicated to the memory of Frantisek Mat\'u\v s}

\begin{document}
\maketitle

\begin{abstract}
The sticky polymatroid conjecture states that any two extensions of the
polymatroid have an amalgam if and only if the polymatroid has no non-modular
pairs of flats. We show that the conjecture holds for polymatroids on five 
or less elements.
\end{abstract}

\section{Introduction}

A polymatroid is \emph{sticky} if any two of its extensions have an amalgam.
This is a direct generalization of the same property of matroids.
If every pair of flats is modular
then the polymatroid is sticky;
the proof in \cite{oxley} generalizes to polymatroids. The
\emph{sticky polymatroid conjecture} states that the converse also holds:
in a sticky polymatroid, each pair of flats is a modular pair.
The corresponding
conjecture for matroids has been stated by Poljak and Turzik
\cite{poljak-turzik}, and received a considerable attention. Poljak
and Turzik showed that the sticky matroid conjecture holds for rank-3
matroids. Bachem and Kern \cite{bachem} showed that the same conjecture
holds in general if it is true for every rank-4 matroid. Generalizing a
result of Bonin in \cite{bonin} which states that a matroid of rank at least
three with 
two disjoint hyperplanes is not sticky, Hochst\"attler and Wilhelmi showed
that matroids having a non-principal modular cut are not sticky
\cite{hoch-wilhelmi}. The same statement for polymatroids was proved in
\cite{convolution} using a convolution-type construction. Thus the sticky
polymatroid conjecture follows from the statement:

\smallskip

\begin{trivlist}\item
\vbox{\par\hangindent=1.6\parindent\hangafter1
$(*)$\space\it
If a polymatroid has a non-modular pair of flats, then it also has a
non-principal modular cut.}
\end{trivlist}
In this note we show that $(*)$ holds for polymatroids on ground set with
at most five elements. Thus a smallest counterexample to the sticky
polymatroid conjecture, if such exists, must have at least six atoms.

Interestingly, the sticky matroid conjecture may not follow from the same
conjecture for polymatroids. The reason is that if two matroids have a
polymatroid amalgam, the value of the rank function of any subset $A$ of the
ground set of either matroid is an integer that is at most $|A|$, the rank
of a set that is not a subset of either ground set might not be an integer
(but can be rational). Interestingly no such example is known. 
At the same time, $(*)$, if true, implies the sticky matroid
conjecture.

\smallskip

All sets in this paper are finite. 
Following the usual practice, ground sets and their subsets are denoted by
capital letters, their elements by lower case letters. The union sign $\cup$
and the curly brackets around singletons are omitted, thus $abA$ denotes the
set $\{a,b\}\cup A$. The \emph{modular defect} of subsets $A$ and $B$ is
defined as
$$
   \delta_f(A,B) = f(A)+f(B)-f(A\cup B) - f(A\cap B).
$$
It is non-negative, and if zero, then the pair $(A,B)$ is called a \emph{modular
pair}.

The paper is organized as follows. Polymatroids, one-point extension, 
and the excess function are defined in Section \ref{sec:definition},
and some basic properties are given. 
Section \ref{sec:linear} introduces the notion of
\emph{linear polymatroid}. This is an intrinsic property shared by all
linearly representable polymatroids. We hope that this notion has further
applications. Two lemmas in Section \ref{sec:main-lemmas} describe
some properties of a minimal counterexample to $(*)$. Using these lemmas and
a property of linear polymatroids, Section \ref{sec:final} shows
that no polymatroid on five or less elements violates $(*)$.


\section{Definitions}\label{sec:definition}

A \emph{polymatroid} $\M=(f,M)$ is a real-valued, non-negative, monotone and 
submodular
function $f$ defined on the set of subsets of the finite set $M$ such
that $f(\emptyset)=0$. Here $M$ is the \emph{ground set}, and $f$ is the
\emph{rank function}. The polymatroid is a \emph{matroid} if all ranks are
integers and $f(A)\le |A|$ for all $A\subseteq M$. For details see
\cite{Lovasz,oxley}. The rank function can be identified with a
$(2^{|M|}-1)$-dimensional real vector, where the indices are the non-empty
subsets of $M$. Vectors corresponding to polymatroids on the ground set $M$
form the pointed polyhedral cone $\Gamma_{\!M}$ \cite{yeungbook}. Its facets are
the hyperplanes determined by the basic submodular inequalities 
$\delta_f(iK,jK)\ge 0$ with distinct $i,j\in M{-}K$ and $K\subseteq M$ 
($K$ can be empty), and the monotonicity requirements $f(M)\ge f(M{-}i)$;
see \cite[Theorem
2]{fmadhe}. Much less is known about the extremal rays of this cone. They
have been computed for ground sets up to five elements \cite{studeny-kocka}
without indicating any structural property. Fixing a polymatroid on each
extremal ray, every polymatroid in $\Gamma_{\!M}$ is a non-negative linear
combination (also called \emph{conic combination}) of these extremal
polymatroids.

\subsection{Flats, modular cuts and filters}

Let $\M=(f,M)$ be a fixed polymatroid. A subset $F\subseteq M$ is a
\emph{flat} if every proper superset of $F$ has strictly larger rank. The
\emph{closure of $A$}, denoted by $\cl(A)$, is the smallest flat containing
$A$. The collection $\F$ of flats is a \emph{modular cut}
if it has properties (i)--(iii) below:
\ITEM{(i)} closed upwards: if $F\in\F$ and the flat $F'$ is a superset of 
$F$, then $F'\in\F$;
\ITEM{(ii)} closed for modular intersection: if $F_1,F_2\in\F$ and
$(F_1,F_2)$ is a modular pair (that is, $f(F_1,F_2)=0$), then 
$F_1\cap F_2\in\F$ (observe that intersection of flats is a flat);
\ITEM{(iii)} not empty, which is equivalent to $M\in\F$.

\smallskip
In standard textbooks, such as \cite{oxley}, the empty collection is also 
considered to be a modular cut. It has been excluded here to emphasize the
similarity to modular filters defined below.

The modular cut \emph{generated by the flats $F_1,\dots,F_k$} is the
smallest modular cut containing all of these sets. This modular
cut is denoted by $\F(F_1,\dots,F_k)$.

A modular cut $\F$ is \emph{principal} if it is generated by a single
flat; or, equivalently, if the intersection of all elements of $\F$
is also an element of $\F$. When $\F$ is not principal, there are two 
flats $F_1, F_2\in\F$ such that $F_1\cap F_2\notin\F$. In this case
$F_1\cap F_2\notin \F(F_1,F_2)$ as the modular cut generated by $F_1$ and
$F_2$ is a subcollection of $\F$.

\smallskip

The collection $\G$ of subsets of $M$ is a \emph{modular filter} if it
satisfies the following properties:
\ITEM{(i)}  closed upwards: if $A\in\G$, $A\subseteq B$, then $B\in\G$;
\ITEM{(ii)} closed for modular intersection: if $A,B\in\G$ and $(A,B)$ is a
modular pair, then $A\cap B\in\G$;
\ITEM{(iii)} non-trivial: if $f(X)=f(M)$, then $X\in\G$.

\smallskip
Modular filters generated by certain subsets as well as principal and
non-principal modular filters can be defined similarly to modular cuts.

The following Proposition shows how to get a modular filter from a modular
cut.

\begin{proposition}\label{prop:cut-to-filter}
Suppose $\F$ is a modular cut. The collection 
$\cl^{-1}(\F)=\{ A\subseteq M: \cl(A)\in\F \}$ is a modular filter.
\end{proposition}
\begin{proof}
Properties (i) and (iii) of the definition of a modular filter
clearly hold, thus suppose $(A,B)$ is a modular pair
and both $F_A=\cl(A)$ and $F_B=\cl(B)$ are in $\F$. Then
$\cl(F_AF_B)=\cl(AB)$ and $F_A\cap F_B\supseteq A\cap B$, thus
\begin{align*}
  & f(F_A, F_B)=f(A, B)+f(A\cap B)-f(F_A\cap F_B)\le{}\\
  & ~ \le f(A,B)=0.
\end{align*}
Consequently $(F_A,F_B)$ is a modular pair and $f(A\cap B)=f(F_A\cap F_B)$,
thus $\cl(A\cap B)=F_A\cap F_B\in\F$, and $A\cap B\in\cl^{-1}(\F)$.
\end{proof}

\subsection{Extensions}

The polymatroid $\M'=(f',M')$ is an \emph{extension} of $\M=(f,M)$ if 
$M'\supset M$, and $f(X)=f'(X)$ for all $X\subseteq M$.
This is a \emph{one-point extension} if $M'{-}M$ has a single element.

Given a polymatroid $(f,M)$ and extensions $(f_1,M_1)$ and $(f_2,M_2)$ of
$(f,M)$ with $M_1{-}M$ and $M_2{-}M$ disjoint, an \emph{amalgam} is a
polymatroid on $M_1\cup M_2$ that extends both $(f_1,M_1)$ and $(f_2,M_2)$.
The polymatroid $\M$ is \emph{sticky} if any two of its extensions
have an amalgam. 

The function $e$ defined on the set of subsets of $M$ is an \emph{excess function of
$\M=(f,M)$} if there is a one-point extension $\M'=(f',M\cup\{x\})$ of $\M$
such that $e(A)=f'(xA)-f'(A)$ for all $A\subseteq M$. If the polymatroid
$\M$ is clear from the context, $e$ is called simply an \emph{excess function}.

\begin{proposition}\label{prop:excess-properties}
The function $e$ is an excess function if and only if the following 
conditions hold.
\ITEM{(i)}
   $e$ is non-negative and decreasing: $e(A)\ge e(B)\ge 0$ for
   $A\subseteq B\subseteq M$,
\ITEM{(ii)}
   $\big(e(M)-e(M{-}i)\big) + \big(f(M)-f(M{-}i)\big) \ge 0$
   for all $i\in M$,
\ITEM{(iii)}
   $\delta_e(aA,bA)+\delta_f(aA,bA)\ge 0$ for all $abA\subseteq M$,
   $a,b\notin A$ (including $A=\emptyset$).
\end{proposition}
\begin{proof}
It is clear that the conditions are necessary. For sufficiency, it is enough
the check that $f'$ defined on subsets of $Mx$ as $f'(Ax)=f(A)+e(A)$ and 
$f'(A)=f(A)$
is the rank function of a polymatroid. According to \cite[Theorem 2]{fmadhe}
it is enough to check $f'(M)\ge f'(M{-}i)$ for $i\in Mx$ and $\delta_{f'}(aA,bA)\ge 0$ 
for $a,b\in Mx{-}A$. These inequalities follow easily from the listed 
conditions.
\end{proof}

The identically zero function clearly satisfies these assumptions, thus
it is an excess function. Actually, it adds a \emph{loop} to the polymatroid.
The inequality $\delta_e(A,B)+\delta_f(A,B)\ge 0$ holds for arbitrary subsets $A, B
\subseteq M$ as this is just the modular defect of the pair $(Ax, Bx)$ in the
extension.

\smallskip

The following statements connect one-point extensions, modular cuts and
modular filters.

\begin{proposition}\label{filter-extension}
The collection $\G$ of subsets of $M$ is a modular filter if and only if there is an
excess function $e$ such that $e(M)=0$ and $\G=\{A\subseteq M: e(A)=0\}$.
\end{proposition}
\begin{proof}
If $e$ is an excess function with $e(M)=0$, then Proposition
\ref{prop:excess-properties} and
 $\delta_e(A,B)+\delta_f(A,B)\ge 0$ trivially imply that $\G$ is a modular
filter. To show the converse, let $\G$ be a modular filter, and choose the
function $e$ as $e(A)=0$ for $A\in \G$, and $e(A)=\eps$ otherwise
where $\eps$ is a sufficiently small positive value. We claim that conditions 
(i)--(iii) of Proposition \ref{prop:excess-properties} hold. This is clear
for (i). For (ii) observe that
$e(M)-e(M{-}i)$ is either $0$ or $-\eps$, and the latter holds when
$M{-}i\notin \G$, but then $f(M{-}i)\not=f(M)$.
Thus choosing $\eps$ smaller than all positive $f(M)-f(M{-}i)$ ensures condition (ii).

Finally, $\delta_e(aA,bA)$ is either non-negative or equals $-\eps$. This latter
happens when both $aA$ and $bA$ are in $\G$ but $A\notin\G$. In this case $(aA,
bA)$ is not a modular pair. Choosing $\eps$ smaller than all possible positive
modular defects in the polymatroid gives condition (iii).
\end{proof}

\begin{claim}\label{claim:modular-cuts}
The following statements are equivalent:
\ITEM{(i)}  The polymatroid $\M$ has a non-principal modular cut.
\ITEM{(ii)} The polymatroid has flats $F_1$ and $F_2$ such that
         $F_1\cap F_2\notin\F(F_1,F_2)$.
\ITEM{(iii)} The polymatroid has flats $F_1$ and $F_2$ and an
       excess function $e$ such that $e(F_1)=e(F_2)=0$ and $e(F_1\cap
       F_2)>0$.
\end{claim}
\begin{proof}
(i) $\rightarrow$ (ii) If $\F$ is not a principal cut, then there are
$F_1,F_2\in\F$ such that $S=F_1\cap F_2\notin\F$. As $\F(F_1,F_2)$ is a
subcollection of $\F$, $S\notin\F(F_1,F_2)$.

(ii) $\rightarrow$ (iii) Let $\F=\F(F_1,F_2)$. By Proposition
\ref{prop:cut-to-filter}, $\G=\cl^{-1}(\F)$ is a modular filter, and by
Proposition \ref{filter-extension} there is an excess function $e$ such that
$e(A)=0$ for $A\in\G$, and $e(A)>0$ otherwise. As $F_1,F_2\in\G$, the first
required property holds, and the second also holds if we show that $S=F_1\cap
F_2\notin \G$. But $S$ is a flat, $S\notin\F$, thus $S\notin\cl^{-1}(\F)=\G$.

(iii) $\rightarrow$ (i) Let $e$ be the excess function, and consider the
collection of flats $\F=\{F: e(F)=0\}$. Clearly, $F_1,F_2\in\F$ and $F_1\cap
F_2\notin\F$. We claim that $\F$ is a modular cut. The fact that it is 
non-principal is clear. Property (i) and (iii)
clearly hold. For (ii) observe that
$\delta_f(A,B)+\delta_e(A,B)\ge 0$, thus if $A,B\in\F$, then $e(A)=e(B)=e(A\cup
B)=0$ (as $e$ is decreasing and non-negative), and if $(A,B)$ is a modular
pair, that is, $\delta_f(A,B)=0$, then
$$
   0\le \delta_f(A,B)+\delta_e(A,B) = -e(A\cap B),
$$
meaning that $e(A\cap B)=0$, thus the flat $A\cap B$ is in $\F$.
\end{proof}


\section{Linear polymatroids}\label{sec:linear}

The polymatroid $\M=(f,M)$ is \emph{linearly representable} if there is a
(finite dimensional) vector space $V$ over some finite field and for each
$i\in M$ a linear subspace $V_i$ of $V$ such that for all $A\subseteq M$,
the rank of $A$ is the dimension of the subspace spanned by $V_A=\bigcup\{
V_i:i\in A\}$; see \cite{oxley}.

A linearly representable polymatroid is clearly integer, and there are
linearly representable polymatroids whose sum is not linearly representable.
Frequently when linearly representable polymatroids have some interesting
(or desired) property, so do polymatroids in their conic hull.
The definition of linear polymatroids below
illustrates such a case. As it captures one of most important aspect of
linear representability, we hope that this notion has other
applications.

\smallskip

Subsets $X$, $Y$ of the ground set are \emph{intersectable} if either they
form a modular pair, or there is an excess function
$e$ such that $e(X)=e(Y)=0$ but $e(X\cap Y)>0$ or, equivalently, if $X\cap
Y\notin\G(X,Y)$, see Proposition \ref{filter-extension}. The polymatroid is
\emph{linear} if every pair of its subsets is intersectable.

\begin{claim}\label{claim:linearly-representable}
Linearly representable polymatroids are linear.
\end{claim}
\begin{proof}
Suppose $\M=(f,M)$ is linearly representable; let $V_i\subseteq V$ be the
linear subspace assigned to $i\in M$. For $A\subseteq M$, its rank is the
dimension of the subspace spanned by $V_A$. If the subsets $X$ and $Y$ are
not modular, adjoin a new element to $\M$ represented by the intersection
of the linear span of $V_X$ and the linear span of $V_Y$. Let $e$ be the
excess function of this one-point extension. Then $e(X)=e(Y)=0$, and
$e(X\cap Y)$ equals the modular defect of $X$ and $Y$, which is non-zero.
\end{proof}

\begin{claim}\label{claim:linear-adds-up}
Conic combination of linear polymatroids is linear.
\end{claim}
\begin{proof}
The definition of linear polymatroids is clearly invariant under
multiplication. So suppose $\M_1$ and $\M_2$ are defined on the 
same ground set and both are linear. If
$(X,Y)$ is modular in $\M_i$, then let $e_i$ be identically zero,
otherwise let it be the excess function guaranteed by linearity. 
If $(X,Y)$ is not modular in $\M_1+\M_2$ then $e_1+e_2$ is the
excess function showing the required extension.
\end{proof}

\begin{claim}\label{claim:linear-is-ok}
Linear polymatroids satisfy $(*)$.
\end{claim}
\begin{proof}
Suppose $(F_1,F_2)$ is a non-modular pair of flats, we need to find a
non-principal modular cut in the polymatroid. As the polymatroid is linear,
there is an excess function $e$ with $e(F_1)=e(F_2)=0$ and $e(F_1\cap
F_2)>0$, and then the existence of non-principal modular cut follows from
Claim \ref{claim:modular-cuts}.
\end{proof}


\section{Main lemmas}\label{sec:main-lemmas}

A non-negative linear (conic) combination of polymatroids on the same set $M$
is again a polymatroid on $M$. If $F$ is a flat in any constituent with
positive coefficient, then $F$ is a flat in the sum; however the sum can have
flats which are not flats in any of the constituents. The next lemmas
establish properties of the constituents when their conic combination violates
$(*)$.

\begin{lemma}\label{lemma:flats}
Let $\M$ and $\N$ be two polymatroids on the same set.
Suppose $\M$ has two flats $F_1, F_2$ such that $F_1\cap
F_2\notin\F_\M(F_1,F_2)$. Then for any $\lambda>0$, $\N+\lambda\M$ satisfies
$(*)$.
\end{lemma}
\begin{proof}
This holds since there is a non-principal modular cut in $\N+\lambda\M$.
This follows from Claim \ref{claim:modular-cuts} once we show that a) $F_1$
and $F_2$ are flats in $\N+\lambda\M$ (this is trivial from the discussion
above), and b) there exists an appropriate excess function $e$ for
$\N+\lambda\M$.

Let $e_\M$ be the excess function for $\M$ with
$e_\M(F_1)=e_\M(F_2)=0$, and $e_\M(F_1\cap F_2)>0$, guaranteed by the condition
and Claim \ref{claim:modular-cuts}, and
define $e=\lambda e_\M$. Conditions in Proposition
\ref{prop:excess-properties} trivially hold (as they are linear), thus $e$
is the required excess function for $\N+\lambda\M$.
\end{proof}

\begin{lemma}\label{lemma:intersecting}
Suppose $\lambda>0$ and $\N+\lambda\M$ is a minimal counterexample to $(*)$. 
In $\M$ every intersecting pair of flats is modular.
\end{lemma}
\begin{proof}
As $\M^*=\N+\lambda\M$ is a counterexample, it has a non-modular pair of
flats but no non-principal modular cut. If $(F_1,F_2)$ is a 
non-modular pair of flats in $\M^*$ and $S=F_1\cap F_2$ is not empty, then 
the contraction $M^*\contract S$ is a smaller counterexample to $(*)$.
Consequently $\M^*$ has no intersecting non-modular flat pairs.

To finish the proof one has to notice that if $F_1$ and $F_2$ are
intersecting non-modular flats in $\M$, then they remain the same in $\M^*$
as well.
\end{proof}

From here the strategy for checking $(*)$ should be clear. Every
polymatroid on a given ground set is a conic combination of finitely many
extremal polymatroids which can be listed explicitly when the polymatroid
has five or less elements \cite{studeny-kocka}. 
Suppose $\M$ violates $(*)$ and no counterexample
exists on a smaller ground set. This $\M$ is a conic combination of the extremal
polymatroids. The combining coefficient is zero if the corresponding
extremal polymatroid a) contains two intersecting non-modular flats (Lemma
\ref{lemma:intersecting}), or b) contains disjoint flats $F_1,F_2$ such that
the modular cut $\F(F_1,F_2)$ is not principal (Lemma \ref{lemma:flats}).
This hopefully leaves only a few extremal polymatroids which can be checked
individually.

\section{Sticky polymatroids on five or fewer elements}\label{sec:final}

\subsection{Polymatroids on two elements}\label{subsec:2}

Let $M=\{a,b\}$. Polymatroids on $M$ are conic combinations of the 
three extremal ones listed in Table \ref{table:2}.
Each of them is linearly representable, thus linear. By Claim
\ref{claim:linear-adds-up} their conic combination remains linear. 
Thus every polymatroid on two elements is linear, and by Claim
\ref{claim:linear-is-ok} they satisfy $(*)$.
\begin{table}[h!tb]%
\begin{center}\begin{tabular}{lccc}
       & $a$ & $b$ & $ab$ \\[2pt]
$\M_a$ &  1  &  0  &  1   \\
$\M_b$ &  0  &  1  &  1   \\
$\M_{ab}$& 1  &  1  &  1
\end{tabular}\end{center}%
\caption{Extremal polymatroids on $\{a,b\}$}\label{table:2}%
\kern -20pt
\end{table}

\subsection{Polymatroids on three elements}\label{subsec:3}

There are eight extremal polymatroids on $M=\{a,b,c\}$. Up to isomorphism
there are four different ones listed in Table \ref{table:3}; the others
can be obtained by permuting the elements of $M$.
\begin{table}[!htb]%
\begin{center}\begin{tabular}{lccccccc}
         & $a$ & $b$ & $c$ & $ab$ & $ac$ & $bc$ & $abc$ \\[2pt]
$\M_a$   &  1  &  0  &  0  &  1   &  1   &   0  &   1    \\
$\M_{ab}$  &  1  &  1  &  0  &  1   &  1   &   1  &   1    \\
$\M_{abc}$ &  1  &  1  &  1  &  1   &  1   &   1  &   1    \\
$\M_*$   &  1  &  1  &  1  &  2   &  2   &   2  &   2
\end{tabular}\end{center}%
\caption{Extremal polymatroids on $\{a,b,c\}$}\label{table:3}%
\kern -10pt
\end{table}
As in the two-element case, all of them are linearly representable,
thus every polymatroid on $\{a,b,c\}$ is linear. By Claim
\ref{claim:linear-is-ok}, $(*)$ holds for these polymatroids. In $\M_*$, every
pair of singletons is independent (modular), but any two determine the
third one.

\subsection{Polymatroids on four elements}\label{subsec:4}

Extremal polymatroids on four and five elements have been reported in
\cite{studeny-kocka}. The software package {\sf Polco} \cite{polco}
can generate the extremal rays from the collection of the defining inequalities.
The polymatroid cone $\Gamma_{\!abcd}$ has 41 extremal rays. There are only
11 different among the corresponding polymatroids up to isomorphism.
Table \ref{table:4} lists one element from each isomorphism
class; the ranks are shown as follows: first one-element subsets,
then two-element subsets, etc., each group in alphabetical order.
\begin{table}[!htb]%
\begin{center}\begin{tabular}{ll}
$\M_1$ & $0,\,0,\,0,\,1,\;0,\,0,\,1,\,0,\,1,\,1,\;0,\,1,\,1,\,1,\;1$\\
$\M_2$ & $0,\,0,\,1,\,1,\;0,\,1,\,1,\,1,\,1,\,1,\;1,\,1,\,1,\,1,\;1$\\
$\M_3$ & $0,\,1,\,1,\,1,\;1,\,1,\,1,\,1,\,1,\,1,\;1,\,1,\,1,\,1,\;1$\\
$\M_4$ & $0,\,1,\,1,\,1,\;1,\,1,\,1,\,2,\,2,\,2,\;2,\,2,\,2,\,2,\;2$\\
$\M_5$ & $1,\,1,\,1,\,1,\;1,\,1,\,1,\,1,\,1,\,1,\;1,\,1,\,1,\,1,\;1$\\
$\M_6$ & $1,\,1,\,1,\,1,\;1,\,2,\,2,\,2,\,2,\,2,\;2,\,2,\,2,\,2,\;2$\\
$\M_7$ & $1,\,1,\,1,\,1,\;2,\,2,\,2,\,2,\,2,\,2,\;2,\,2,\,2,\,2,\;2$\\
$\M_8$ & $1,\,1,\,1,\,1,\;2,\,2,\,2,\,2,\,2,\,2,\;3,\,3,\,3,\,3,\;3$\\
$\M_9$ & $1,\,1,\,1,\,2,\;2,\,2,\,2,\,2,\,2,\,2,\;2,\,2,\,2,\,2,\;2$\\
$\M_{10}$ & $1,\,1,\,1,\,2,\;2,\,2,\,3,\,2,\,3,\,3,\;3,\,3,\,3,\,3,\;3$\\
$\M_{11}$ & $2,\,2,\,2,\,2,\;3,\,3,\,3,\,3,\,3,\,4,\;4,\,4,\,4,\,4,\;4$
\end{tabular}\end{center}%
\caption{Extremal polymatroids on $\{a,b,c,d\}$}\label{table:4}%
\kern -10pt
\end{table}
Polymatroids $\M_1$--$\M_{10}$ are linearly representable. For $\M_{10}$
take three linearly independent vectors $\mathbf a$, $\mathbf b$, $\mathbf
c$; they span the one-dimensional subspaces assigned to $a,b,c$,
respectively, while $d$ gets the 2-dimensional subspace spanned by ${\mathbf
a}+\mathbf{b}$ and ${\mathbf a}+{\mathbf c}$. The polymatroid $\M_{11}$ is not linearly
representable, but the modular cut generated by the flats $ac$ and $bd$ is
$\{ac,bd,abcd\}$. Thus by Lemma \ref{lemma:flats} it cannot contribute to a
$(*)$-violating polymatroid. Consequently each polymatroid on four elements
satisfies $(*)$.

\subsection{Polymatroids on five elements}\label{subsec:5}

According to \cite{studeny-kocka} there are 117983 extremal polymatroids on
a five element set. Up to isomorphism there are 1320 different ones. By Lemmas
\ref{lemma:flats} and \ref{lemma:intersecting} we can eliminate those
extremal polymatroids that a) contain two intersecting non-modular flats,
or b) contain disjoint flats $F_1,F_2$ such that $\F(F_1,F_2)$ is not
principal (i.e., the generated cut does not contain the empty set). After this thinning 
\begin{table}[!htb]%
\thinmuskip=2mu%
\begin{center}\begin{tabular}{c}
$0,0,0,0,1,0,0,0,1,0,0,1,0,1,1,0,0,1,0,1,1,0,1,1,1,0,1,1,1,1,1$\\
$0,0,0,1,1,0,0,1,1,0,1,1,1,1,1,0,1,1,1,1,1,1,1,1,1,1,1,1,1,1,1$\\
$0,0,1,1,1,0,1,1,1,1,1,1,1,1,1,1,1,1,1,1,1,1,1,1,1,1,1,1,1,1,1$\\
$0,0,1,1,1,0,1,1,1,1,1,1,2,2,2,1,1,1,2,2,2,2,2,2,2,2,2,2,2,2,2$\\
$0,1,1,1,1,1,1,1,1,1,1,1,1,1,1,1,1,1,1,1,1,1,1,1,1,1,1,1,1,1,1$\\
$0,1,1,1,1,1,1,1,1,1,2,2,2,2,2,1,2,2,2,2,2,2,2,2,2,2,2,2,2,2,2$\\
$0,1,1,1,1,1,1,1,1,2,2,2,2,2,2,2,2,2,2,2,2,2,2,2,2,2,2,2,2,2,2$\\
$0,1,1,1,2,1,1,1,2,2,2,2,2,2,2,2,2,2,2,2,2,2,2,2,2,2,2,2,2,2,2$\\
$1,1,1,1,1,1,1,1,1,1,1,1,1,1,1,1,1,1,1,1,1,1,1,1,1,1,1,1,1,1,1$\\
$1,1,1,1,1,1,1,2,2,1,2,2,2,2,2,1,2,2,2,2,2,2,2,2,2,2,2,2,2,2,2$\\
$1,1,1,1,1,1,2,2,2,2,2,2,1,2,2,2,2,2,2,2,2,2,2,2,2,2,2,2,2,2,2$\\
$1,1,1,1,1,1,2,2,2,2,2,2,2,2,2,2,2,2,2,2,2,2,2,2,2,2,2,2,2,2,2$\\
$1,1,1,1,1,2,2,2,2,2,2,2,2,2,2,2,2,2,2,2,2,2,2,2,2,2,2,2,2,2,2$\\
$1,1,1,1,2,1,2,2,2,2,2,2,2,2,2,2,2,2,2,2,2,2,2,2,2,2,2,2,2,2,2$\\
$1,1,1,1,2,2,2,2,2,2,2,2,2,2,2,2,2,2,2,2,2,2,2,2,2,2,2,2,2,2,2$\\
$1,1,1,2,2,2,2,2,2,2,2,2,2,2,2,2,2,2,2,2,2,2,2,2,2,2,2,2,2,2,2$\\
$1,1,1,2,2,2,2,2,2,2,3,3,3,3,3,2,3,3,3,3,3,3,3,3,3,3,3,3,3,3,3$\\
\end{tabular}\end{center}%
\caption{Remaining extremal polymatroids on $\{a,b,c,d,e\}$}\label{table:5}%
\kern -10pt
\end{table}
we get a quite meager set of 17 isomorphism classes; representatives 
are listed in Table \ref{table:5}.
The ranks are shown by cardinality of the subset, and within that 
alphabetically. By inspection, all of them are linearly representable, thus
linear. By Claim \ref{claim:linear-adds-up}, any conic combination of linear
polymatroids is linear, and Claim \ref{claim:linear-is-ok} says that these 
linear polymatroids satisfy $(*)$. Consequently all polymatroids on 5 elements
satisfy $(*)$.



\section*{Acknowledgment}
The research reported in this paper was supported by GACR project
number  19-04579S, and partially by the Lend\"ulet program of the HAS.

\end{document}